\documentclass[11pt]{amsart}

\usepackage{amsaddr}

\usepackage{amsthm}
\usepackage{amssymb}
\usepackage{amsmath}

\newtheorem{theorem}{Theorem}
\newtheorem{lemma}[theorem]{Lemma}
\usepackage{enumerate}

\newcommand{\out}{\mathrm{out}}
\newcommand{\expose}{\mathrm{expose}}

\begin{document}

\title[Deviator Detection under Imperfect Monitoring]{Deviator Detection 
under Imperfect Monitoring \\[1ex] 
\footnotesize\mdseries
Extended Abstract}

\author[D. Berwanger]{Dietmar Berwanger}
\address{{\smaller CNRS and ENS Paris-Saclay, France}}  
\email{dwb@lsv.fr (D. Berwanger)}
\author[R. Ramanujam]{R. Ramanujam}
\address{The Insitute of Mathematical Sciences, India}
\email{jam@imsc.res.in (R. Ramanujam)}

\maketitle

\begin{abstract}
Grim-trigger strategies are a fundamental mechanism 
for sustaining equilibria in iterated games: 
The players cooperate along an agreed path, and as soon as one
player deviates, the others form a coalition to play him down to his
minmax level. A precondition to triggering such a strategy is that the
identity of the deviating player becomes common knowledge among the
other players. This can be difficult or impossible to attain in games
where the information structure allows only imperfect monitoring of
the played actions or of the global state.   

We study the problem of synthesising finite-state strategies for
detecting the deviator from an agreed strategy profile in games played
on finite graphs with different information structures. We show that
the problem is undecidable in the general case where the global state
cannot be monitored.
On the other hand, we prove that under perfect monitoring of the
global state and  
imperfect monitoring of actions, the problem becomes
decidable, and we present an effective synthesis procedure 
that covers infinitely repeated games with private monitoring. 
\end{abstract}

\section{Introduction}
In social situations, a {\em queue} acts in a self-stabilising manner: when
anyone tries to jump the queue, the others give him a dirty look, and often
this suffices to enforce the rule. Distributed protocols are often designed
to be self-stabilising in a similar sense: when a fault occurs, it is detected
and isolated, and the system recovers to a state in such a way that the
computation can proceed.

In general, the design of such systems assumes that all processes cooperate 
and coordinate their actions towards achieving the system goals. When a process
deviates from the protocol (perhaps due to faults), the system needs to detect
it and recover from the situation. Usually, if such deviation occurs in an
isolated one-off manner, it may be hard to detect, but when they occur
repeatedly, there is the possibility of detecting the culprit(s). This idea has 
been applied to intruder detection in security theory.

There are several interesting variations on this theme. One is to ask for 
protocols or solutions that do not demand detection of deviators but to
merely provide resilience. Typical solutions in distributed computing
assume a bound on the number of faulty processes and provide solutions that
can tolerate that many failures. Another consideration relates to the 
observational limitations of processes. In a distributed system, processes
have only a partial view of the system state, and often cannot observe all
the moves of other processes, and these may help the deviator to evade
detection successfully.

When distributed computing meets game theory, we have further interesting
possibilities to consider. 
A player may act selfishly if it maximises her payoff, 
even if this means a deviation.  
In general the coalition may have no way to hold 
a member of the coalition to act according to prior commitments 
(dismissively labelled ``cheap talk''). 
However, in the case of repeated play, there can be threats and 
punishment to ensure that members do not deviate. In this context, a variety of 
mechanisms are studied in game theory, notable among them the {\em grim-trigger 
strategies}: start out co-operating, and when any player deviates, punish for 
ever after. The importance of such strategies is that they play a central role 
in what are referred to as Folk theorems in game theory.

Nash equilibria provide a robust way to predict how rational players would act 
offering their best responses to their beliefs about how others might act, in 
situations where players differ in their knowledge of what others might
be doing. At Nash equilibrium, no player has an incentive to deviate unilaterally 
and shift to a different strategy. 
Folk theorems then assert that for any Nash 
equilibrium outcome vector $(o_1, \ldots, o_n)$ in an $n$-player infinitely 
repeated game (with, say, limit average rewards), then each player $i$ can 
force the outcome $o_i$. Conversely every ``feasible'' and enforceable outcome 
is the outcome of some Nash equilibrium.

Interestingly, these (and many related) results depend on the ability of players
to perfectly monitor the actions of the deviator. In cases where players' 
actions
may not be directly observable, special solutions are needed, and game theorists
have developed a powerful set of techniques for many subclasses of games~\cite{FudenbergLM94,KandoriHit98,Amarante03,MailathSam06}.
Note that this situation pertains more directly to distributed systems
where processes are limited in their ability to observe the global state and to
record other processes' actions.

In this paper, we focus on computational questions related to deviator detection 
under imperfect monitoring. Can we have algorithms that determine, in a game, whether 
the deviator's identity be rendered common knowledge ? If yes, we could like to 
construct strategies (for the other players) to achieve this.  This problem is 
related to solving consensus problems on graphs with different communication topologies.

We propose a general technique based on methods from automated synthesis
towards achieving an epistemic objective (involving common knowledge). With
imperfect information, the synthesis problem is undecidable already with simple
reachability objectives. So one expects only negative results for deviator detection 
as well. And indeed, deviator detection under imperfect information of game
state is in general undecidable.

However, it is interesting to note that the essence of the deviator detection
problem lies in monitoring of player actions and not necessarily game state.
Indeed, in repeated games we often have no states at all ! We show in this
paper that in such as case, the problem is tractable. The main idea is that the
problem can then be studied in the setting of coordination games of incomplete
information with finitely many states of nature. These are systems with perfect
monitoring of state, and uncertainty for a player comes only from unobserved
actions of other players. But then, these are games with bounded initial uncertainty
that can only reduce as play progresses, and this observation leads us to an
algorithmic solution.

While the main result of the paper is the assertion that deviator detection
is decidable under imperfect monitoring of actions (with only a finite amount of
uncertainty about the state), we see the contribution of the paper as twofold: to
highlight a setting in game theory that is of interest to distributed computing;
and to illustrate the use of epistemic objectives, that is of interest to games
as well as distributed systems. We also suggest that methods from automated
synthesis may offer new ways of describing sets of equilibrium solutions and for
constructing equilibria, which could be of technical interest for bridging game
theory and distributed systems~\cite{halpern-csgt}.

\section{Games}
We model distributed systems as infinite games with finitely many states. 
There is a finite set $I$ of players. 
We refer to a list $x=(x^i)_{i \in I}$ with one element $x^i$ for every
player~$i$ as a \emph{profile}.  
For any such profile, we write $x^{-i}$ to denote the 
list $( x^j )_{i \in I , j \neq i}$ where the component
of Player~$i$ is omitted; 
for element $x^i$ and a list $x^{-i}$, we denote by
$(x^i, x^{-i} )$ the full profile $(x^i)_{i \in I}$. 

\paragraph{Game structure.}
To describe the game dynamics, we fix, for each player~$i \in I$, a~set~$A^i$ of
\emph{actions}, a set~$B^i$ of \emph{observations}, and a set $V^i$ of
\emph{local states}\,---\,these are finite sets. We denote by $A$, $B$, and
$V$, the set of all profiles of actions, observations, and local
states; a profile of local states is also called \emph{global state}.
Now, the game form is described by its transition function 
$\gamma: V \times A \to V \times B$. 

The game is played in stages over infinitely many periods starting from a 
designated initial state $v_0 \in V$ 
known to all players.  In each period~$t \ge 1$, starting in a state~$v_{t-1}$, 
every player~$i$ chooses an action $a^i_t$.  Then the transition 
$\gamma( v_{t-1}, a_t) = (b_t, v_t)$ determines the observation $b_t^i$ received 
privately by each player~$i$, and the  the global successor state $v_t$, from which
the play proceeds to period $t+1$.

Thus, a~\emph{play} is an infinite sequence $\pi = v_0 a_1 v_1 a_2 v_2 
\dots\in V(AV)^\omega$ following the transitions $\gamma(v_{t-1}, a_t) = 
(v_t, b_t)$ for all $t \ge 1$.  A \emph{history} is a finite prefix 
$v_0 a_1 v_1 ~\dots a_t v_t \in V (AV)^*$ of a play.  We refer to the number 
of stages played up to period~$t$ as the \emph{length} of the history.
The sequence of observations $b_1^i b_2^i \dots b_t^i$ received by
player~$i$ along a history $\pi$ is denoted by $\beta^i( \pi )$.
We assume that each player~$i$ always knows his local state $v^i$
and the action $a^i$ she is playing, that is, these data are included 
in her private observation received in each round.  However, she is not 
perfectly informed about the local states or the actions of the other players, 
therefore we speak of \emph{imperfect monitoring}. 
The monitoring function~$\beta^i$ induces an \emph{indistinguishability} 
relation between histories and plays: $\pi \sim^i \pi'$ if, and only if, 
$\beta^i( \pi ) = \beta^i (\pi')$. This is an equivalence relation 
between game histories; its classes are called the 
\emph{information sets} of player~$i$.  

A \emph{strategy} for player~$i$ is a mapping $s^i:(B^i)^* \to A^i$
that prescribes an action for every observation sequence. 
Again, we denote by $S$ the set of all strategy profiles.
We say that a history or a play $\pi$ \emph{follows} a strategy~$s^i$, if 
$a^i_{t+1} = s^i( \beta^i( \pi_t ))$, for all histories $\pi_t$ of 
length $t \ge 0$ in $\pi$.  Likewise, a history or play follows a profile 
$s \in S$, if it follows the strategy $s^i$ of each player~$i$.  
The \emph{outcome} $\mathrm{out}( s )$ of a strategy profile~$s$
is the unique play that follows it. 
 
With the above definition of a strategy, we implicitly assume that players have
perfect recall, that is, they may record all the information acquired
along a play. Nevertheless, in certain cases, we can restrict our
attention to strategy functions computable by automata with finite memory. 
In this case, we speak of \emph{finite-state strategies}.

\paragraph{Strategy synthesis.}
The task of automated synthesis is to construct finite-state strategies for 
solving games (presented in a finite way).
Depending on the purpose of the model, the notion of solving has different 
meanings.

One prominent application area in distributed systems is concerned with
synthesising coordination strategies for a coalition with 
common interests against a fixed adversary\,---\,the environment, 
or Nature~\cite{KupfermanVar01,PnueliRos90,FinkbeinerSch05}.
For this purpose, we assume
the coalition to be the set $I \setminus \{0\}$ excluding a designated player~$0$.
We are interested in win/lose games. The {\em winning condition}
is described as a set $W$ of plays; a basic example are 
\emph{reachability} winning conditions, which consist of all plays that reach a designated 
set of global states.
Here, a solution is a distributed winning strategy for the coalition: 
a profile $s^{-0}$ such that $\out( s^{-0}, r^0 ) \in W$ for 
all $r^0 \in S^0$.

The \emph{distributed synthesis problem} for coordination strategies 
asks: for a given game, determine whether there 
exists a profile of finite-state strategies for $I \setminus \{0\}$ 
that is winning against player~$0$. 
This problem is well known to be undecidable, 
already for games with reachability winning conditions.

\begin{theorem}[\cite{PetersonRei79}]\label{thm:coordinationundecidabilty}
  The distributed synthesis problem is undecidable for 
  reachability games for two players with imperfect information against an adversary.
\end{theorem}

If we consider non-zero sum games,
where we have players with different, possibly overlapping
objectives, a standard solution concept is that of
Nash equilibrium. A theory of synthesis for such games
is being developed in the last decade {\cite{GutierrezWoo14,ChatterjeeEtAl14,BBMU15}}.
While there are several positive results for
games of perfect information, synthesis questions are
difficult in the context of imperfect information. 

In this paper, we are interested in a specific case that lies between 
the approaches of distributed coordination and Nash equilibrium:  
synthesis 
of strategies for detecting an unknown adversary\,---if he should arise---\, 
under conditions of imperfect monitoring. 
We now proceed to define and address this problem.

\section{Deviator detection}

Unlike the traditional setting concerned with \emph{temporal}
objectives on actions or states assumed in a play, we are interested here in 
\emph{epistemic} objectives which refer to attaining knowledge 
that a certain event has occurred.
For an introduction to knowledge in distributed systems, 
we refer to the book of 
Fagin, Halpern, Moses, and Vardi~\cite[Ch. 2]{FaginHMV95}.

Let us fix a game~$\gamma$ with the usual notation.
An \emph{event} is a subset~$E$ of 
histories in~$\gamma$. The event~$E$
\emph{occurs} at history $\pi$ if $\pi \in E$. 
We say that (the occurrence of) $E$ is private knowledge 
of Player~$i$ at history $\pi$, if 
for any $\pi' \sim^i \pi$, it holds that $\pi' \in E$.
Further, an event $E \subseteq \Omega$ is \emph{common knowledge} among the players
of a coalition~$C \subseteq I$ at history $\pi$ 
if, for every sequence of histories $\pi_1, \dots, \pi_k$ and players 
$i_1, \dots, i_k \in C$ such that 
$\pi \sim^{i_1} \pi_1 \sim^{i_2} \dots \sim^{i_k} \pi_k$,
it is the case that $\pi_k \in E$. 

Specifically, we are interested in the event that a player~$i$ 
has deviated from a given play.
To describe this, we define for each play~$\pi$ and for every player~$i \in I$,
the event $D^i( \pi )$
consisting of all histories $\tau$ that disagree with $\pi$ such that,
for the first round $t$ where the prefixes $\pi_t$
and $\tau_t$ disagree, they differ only in the action of player~$i$.
(Since the transition
function is deterministic and the initial state is fixed, the first difference
between two histories can only occur at an action profile). 
Obviously, the sets $D^i( \pi )$ are suffix closed, that is, for each
$\tau \in D^i( \pi )$, any prolongation history $\tau'$ belongs to
$D^i( \pi)$ as well. Likewise, if a coalition 
attains common knowledge of $D^i( \pi )$ at a history $\tau$, it
attains common knowledge of $D^i( \pi )$ at every prolongation
history of $\tau$. 

Now, let $W$ be a designated set of plays in $\gamma$. 
A \emph{deviator detection} strategy with respect to~$W$ is a strategy
profile $s \in S$ such that:
\begin{enumerate}[(i)]
\item the outcome $\pi = \out( s )$ belongs to $W$, and 
\item for each player~$i$ and every strategy $r^i \in S^i$, 
if the outcome $\pi' = \out( s^{-i}, r^i )$ disagrees with $\pi$, 
then the coalition $I \setminus \{ i\}$ attains common knowledge of $D^i ( \pi )$ at some
history of $\pi'$.
\end{enumerate}

The synthesis problem for deviator detection is the following:
given a game~$\gamma$ with a target set $W$ specified by a finite-state
automaton, decide whether there exists a finite-state strategy profile for
deviator detection with respect to~$W$ and, if so, construct one.

\subsection{Deviator detection as a coordination problem}\label{ssec:detection-coordination}

Alternatively, we can cast the deviator detection problem 
as a more standard problem of distributed synthesis with temporal objectives.
Informally, this is done by adding a new player\,---Nature---\,
that can either remain silent, or at any point take over 
the identity of an actual player, deviate from his intended action, 
and continue playing on his behalf. 
The deviation of Nature takes the game into a fresh copy associated to the corrupted player~$i$ 
where the only way to win for the remaining coalition 
is by issuing a simultaneous action 
in which they all expose the identity of~$i$; however, 
if the exposure action is not taken in 
consensus by all players, except for~$i$, the game is lost.
  
More precisely, we transform the deviator detection game~$\gamma$ 
into a coordination game~$\hat{\gamma}$ against Nature\,---let us call it exposure game---\,as
follows: First, we add for each player~$i$, actions~$j \neq i$ that allow to
expose a deviation by player~$j$, that is, 
we set $\hat{A}^i := A^i \cup \{I \setminus i\}$; we use the shorthand $\expose^i$ 
to denote any action profile where the coalition $I \setminus \{ i \}$ chooses~$i$ in consensus.
Further, we involve a new player $N$ with actions
that allow him to stay silent (by choosing~$\varepsilon$) 
or to corrupt the action of any other player, that is
$A^N := \{ \varepsilon \} \cup \cup_{i \in I} A^i$. His local states are
$V^N := \{\varepsilon\} \cup I$. Moreover, we
include (global) sink states $\oplus$ and $\ominus$. 
The observation sets remain unchanged.
The transitions of the new game~$\hat{\gamma}$ follow the original 
transitions as long as Nature stays silent: $\hat{\gamma}( (v, \varepsilon), (a, \varepsilon) ) = ((v',
\varepsilon), b)$ for $(v', b) := \gamma(v, a )$. 
When Nature decides to deviate from the intended 
action of a player~$i$, his local state changes from $\varepsilon$ to $i$: 
$\hat{\gamma}( (v, \varepsilon), (a, c^i) ) = ((v',
i), b)$ for $c^i \neq a^i$ and $(v', b) := \gamma(v, (a^{-i}, c^i) )$; 
for the transitions in the game copy where player~$i$ is corrupted, we set  
$\hat{\gamma}( (v, i), (a, c^i) ) = ((v', i),
(a, c^i) )$ for $(v', b) := \gamma(v, (a^{-i}, c^i) )$, and 
$\hat{\gamma}( (v, i), (\expose^i, c^i) ) = \oplus$; all other moves involving exposure actions 
lead to $\ominus$. The (temporal) winning condition~$\hat{W}$ of the new game consists of the plays 
in $W$ where Nature remained silent and of the plays that reach $\oplus$. 

We can analyse the exposure game in terms of the Knowledge of Preconditions principle formulated by Moses in~\cite{Moses16}. 
Once a deviation occurred, the coalition can win only by reaching $\oplus$ via a simultaneous consensus action which requires common knowledge 
of the identity of the deviator. 
Conversely, deviator detection strategies in $\gamma$ 
can be readily used to win the exposure game. 

\begin{lemma}
Every coordination strategy for the coalition~$I$ in the game $(\hat{\gamma}, \hat{W})$ 
against~$N$ corresponds to a deviator-detection strategy with respect to $W$ 
in $\gamma$, and vice versa.
\end{lemma}

\subsection{General undecidability}

The translation of deviator detection games into coordination games 
shows the problem under a different angle, 
but it does not bring us closer to an algorithmic solution. 
In the general setting of imperfect monitoring, we obtain coordination games between multiple players
with imperfect information, for which the synthesis problem is undecidable,
as pointed out in  Theorem~\ref{thm:coordinationundecidabilty}.

Indeed, it turns out that under imperfect monitoring, detecting deviators is no
easier than coordinating against an opponent to reach a target set.

\begin{theorem}
The synthesis problem for deviator detection strategies is undecidable for games with 
imperfect monitoring. 
\end{theorem}

\begin{proof}
Consider an arbitrary coordination game~$\gamma$ with three players~$0$, $1$, and~$2$ 
where the coalition $\{1, 2\}$ seeks to reach a set $F$ of states under imperfect monitoring. 
We reduce the synthesis problem for this game to one in a deviator detection game~$\gamma'$ among four 
players $1$, $2$, $X$ and $Y$, in which $1,2$ play the same role as in $\gamma$ 
whereas $X$ and $Y$ both take the role of player~$0$. The new game contains two disjoint copies 
$\gamma_X$, $\gamma_Y$ of $\gamma$; the actions of $Y$ are ignored in the former, 
and those of $X$ in the latter. The game~$\gamma'$ starts in a fresh state 
at which it loops with a fixed action profile; 
the designated set~$W$ consists only of this looping play.
The actions of player~$1$ and $2$ at this state are perfectly observable to all players, 
so any deviation from the loop in~$W$ is detected instantly.
In contrast, the deviations of player~$X$ or $Y$ generate the same (fresh) observation to $1$ and $2$, 
and they lead to the initial state of $\gamma_X$ and $\gamma_Y$, respectively. 
These two component games evolve in the same way with the only difference that,
 when switching to a target state of $F$ in $\gamma_X$, 
the observation $X$ 
is sent to all players, whereas the observation~$Y$ is sent when reaching~$F$ in $\gamma_Y$. 

Thus, for any deviator detection strategy~$s$, upon deviation of either $X$ or $Y$ from the initial loop, 
players~$1$ and $2$ 
must coordinate to reach the target set~$F$ to identify the deviator. 
Hence, $(s^1, s^2)$ yields a solution of the coordination problem. 
Conversely, any solution of the coordination problem leads to a state in~$F$ where the deviator is revealed, so it provides a deviator detection strategy. Since the synthesis problem for coordination problems is undecidable, according to Theorem~\ref{thm:coordinationundecidabilty}, 
it follows that the synthesis problem for deviator detection strategies is also undecidable.  
\qed 
\end{proof}

\section{Perfect monitoring of states}\label{Sec:Main}

As we could see in the previous section, the algorithmic intractability of
games where the global state can be hidden from the players 
over an unbounded duration of time is preserved when we move from 
coordination to the deviator detection problem. 
However, our setting bears two sources of uncertainty: the global state \emph{and} 
the played action. 
In this section we consider the case where the uncertainty comes only 
from the actions played by the other players. 
Indeed, this is a generalisation of the setting of infinitely repeated games, 
which can be seen as games with only one global state.


As an example, consider the following simple variant of a 
beeping model~\cite{CornejoKuhn2010}. There are $n$ nodes in a network 
represented by an undirected graph. The nodes can communicate synchrounously. 
In every round, a node can either beep or stay silent. 
A silent node can observe 
whether at least one of its neighbour beeped. 
We assume that the network is commonly known, we are interested in 
distributed protocols under wich some temporal condition is ensured,  e.g., 
no more than a quarter of the nodes beep in the same round, and that 
are additionally deviator proof, in the sense that   
whenever a node deviates from 
the protocol, the protocol followed by the remaining nodes allows to 
reach a consensus on the identity of the deviator. 
This question can be  
represented as a deviator-detection problem among $n$ 
players, each with two actions -- beep or stay silent -- and 
two observations, telling whether any neighbour beeped or not in the 
current round. As the effect of an action profile is the same in any round, 
the game has only one global state. 
Still, the synthesis problem shows some complexity.
Partly, this is due to the structure of the observation functions 
encoded by the network graph. For instance, one can observe 
that no deviator detection-strategy can exists on networks that are 
not two-connected: Any deviation has to be detected by at least 
two witnesses, and every node that is not a direct witness needs to be 
finally informed via at least two disjoint paths. 
But the greater challenge comes from 
the dynamics of communicating the identity of the deviator: 
In contrast to the more traditional synthesis problems for temporal 
conditions, whether a play~$\pi$ is successful is not determined by 
the strategic choices taken along~$\pi$ itself, but also depends on 
the choices taken on histories connected to~$\pi$ via the player's 
indistinguishability relations. 
 
Concretely, we consider games that allow \emph{perfect monitoring of the state} 
in the sense that 
for every observation sequence $\beta^i := b_1^i, \dots, b_t^i$ received by any player~$i$ 
along a history, there exists precisely one global state $v$ 
that is reachable by a history with observation
$\beta^i$. In other words, all histories in an information set of a player
end at the same global state. 
The condition is obviously met if we include the current global state in the observation 
of each player. Indeed, every finite game with perfect monitoring of the state 
can be transformed effectively into one where all the players can observe 
the current state.   

Our main technical result establishes that,  
under perfect monitoring of states,
the deviator detection problem is algorithmically tractable in spite of imperfect private monitoring 
of actions.
 
\begin{theorem}\label{thm:main}
The synthesis problem for deviator-detection strategies is effectively solvable for
games with perfect monitoring of the state and imperfect monitoring of actions. 
\end{theorem}

The proof relies on a more general result which states, informally, 
that the synthesis problem for coordination games with a bounded amount 
of uncertainty are algorithmically tractable.
To formulate this more precisely, let us fix a set~$I$ of players 
with their sets of actions, observations, and local states; 
the set includes a designated player~Nature.
Consider a finite collection $\gamma_1, \dots, \gamma_n$ of games with 
perfect monitoring of the state over the fixed action, 
observation and state spaces, together with a winning condition~$W$ 
common to all these games. 
We define the \emph{sum game} $\gamma$ over the collection as a game with
a fresh initial node at which  Nature 
chooses the initial node of any of the games in the collection; no information about this
move is delivered to the other players in~$I \setminus \{ \mathrm{Nature} \}$. 
Note that the sum~$\gamma$
may not allow perfect monitoring of the state. 
For this sum game~$\gamma$, we consider the task to synthesise a coordination strategy for the 
coalition~$I \setminus\{ \mathrm{Nature} \}$ to ensure either that the outcome is 
either winning with respect to $W$ or it \emph{reveals} the initial choice of Nature. 
That is, we require, for every play~$\pi \not \in W$ 
which follows the strategy, that at some history in~$\pi$, the players attain 
common knowledge about the component game that has been chosen. 
We call this a revelation game over $\gamma_1, \dots, \gamma_n$ with condition~$W$.

In game-theoretic terminology, the sum game constructed above is actually a game of 
\emph{incomplete} information \,---\, the uncertainty about the global state of the game is 
due to not knowing  which of the finitely many component games is being played. 
Nevertheless, as the component games deliver different observations, 
the players may be able to recover this missing information. It turns out that this
restricted form of imperfect information is algorithmically tractable. 
The setting is similar to that of multi-environment Markov decision processes 
studied in \cite{RaskinSan14}. 

\begin{theorem}\label{thm:revelation}
The synthesis problem is effectively solvable for revelation games on  
components with perfect monitoring of the state.
Moreover, the set of all 
winning strategies admits a regular representation. 
\end{theorem}

The idea is to keep track of the knowledge 
that players have about the index of the actual component game 
while the play proceeds. 
This knowledge can be represented by epistemic 
structures similar to the ones used in~\cite{BKP11}.
Here it is sufficient to consider epistemic structures on a subset
of component indices, with epistemic equivalence relations~$\sim^i$ 
relating indices $k$, $k'$ whenever player~$i$ considers it possible to be in
component~$k$ if the actual history is in component~$k'$. 
The tracking construction associates to each history~$\pi$
a structure that is strongly connected 
via these relations; we call this structure the epistemic state of $\pi$.  

Intuitively, the construction represents 
the actions in the original game abstractly by their 
effect on the uncertainly about the component. 
In contrast to the concrete actions in the game, which can be monitored only imperfectly, 
the abstract updates on the epistemic structure can be 
monitored perfectly; the resulting game is thus solvable with standard 
methods as one with perfect information, 
where the winning condition asks to satisfy~$W$ or to reach an 
epistemic structure over a singleton, representing that the 
players attain common knowledge about the actual component.  
The perfect-information solution 
yields a regular representation of the set of winning strategies 
over a product alphabet of global game states and epistemic states. 
In the full paper, we
show that this abstraction can be maintained by a
finite-state construction and that it allows to represent a solution 
whenever one exists. 

To prove Theorem~\ref{thm:main} using the result of Theorem~\ref{thm:revelation}, 
we view the deviator detection problem as a revelation game. Towards this, we adapt
the exposure game constructed in
Subsection~\ref{ssec:detection-coordination} for a given deviator detection problem~$\gamma$ with target set~$W$.  
The exposure game $(\hat{\gamma}, \hat{W})$ 
is already close to the setting of revelation games, 
but there is one twist: 
Besides choosing the deviator, in the exposure game 
Nature can also choose the period in which to deviate. 
To account for this, we transform~$\hat{\gamma}$ by letting 
Nature pick a candidate deviator $i \in I$ in the first move; 
this choice is hidden from the other players. 
In every later round, Nature can choose to either remain silent or 
to corrupt the action of player~$i$. As a target condition for the new game, 
we fix the set of all plays in the target set~$W$ where Nature remained silent. 
The obtained revelation game has the same set of solutions as the 
deviator detection problem at the outset.  

\section{Conclusion}

Thus what we have here is a building block for constructing equilibria in games based on 
epistemic states. We are interested primarily in the issue of detecting a 
deviation from an agreed strategy profile. This task is more specific than 
constructing equilibria by detecting deviations 
from the set of distributed strategies that ensure a win, as done 
for instance, in~\cite{BouyerBMU12}. The crucial difference relies in the fact 
that, in the latter case, the deviation events can be described by a 
(regular) set of game histories, while in our setting, the notion of deviation 
is relative to a strategy profile that is not fixed within the game structure. 
To illustrate the difference, 
consider the example of a beeping model from~\ref{Sec:Main} 
with the trival target set that contains all possible plays. 
Obviously, every strategy profile is an equilibrium here, but 
deviator-detection strategies remain nevertheless intricate.

Our abstraction from imperfect monitoring of actions to games with perfect
information is fairly generic. 
The central clue is that if there is only a 
finite amount of information hidden in the beginning of a play, and 
we can decide whether the 
coalition can recover it. 
In the context of distributed systems, there is a wide variety
of situations that involve only imperfect observation of actions, and where uncertainty
about system states may be bounded. Hence we reasonably expect that these techniques
will be applicable, not only for deviator detection, but in other algorithmic questions
on inferring global information in such systems.

\bigskip
\noindent\paragraph{Acknowledgements} 
This work was supported by the 
Indo-French Joint Research Unit \textsc{ReLaX}
(\textsc{umi} \textsc{cnrs} 2000).

\bibliographystyle{amsplain}
\bibliography{global}

\providecommand{\bysame}{\leavevmode\hbox to3em{\hrulefill}\thinspace}
\providecommand{\MR}{\relax\ifhmode\unskip\space\fi MR }
\providecommand{\MRhref}[2]{%
  \href{http://www.ams.org/mathscinet-getitem?mr=#1}{#2}
}
\providecommand{\href}[2]{#2}
\begin{thebibliography}{10}

\bibitem{Amarante03}
Massimiliano Amarante, \emph{Recursive structure and equilibria in games with
  private monitoring}, Economic Theory \textbf{22} (2003), no.~2, 353--374.

\bibitem{BKP11}
Dietmar Berwanger, {\L}ukasz Kaiser, and Bernd Puchala, \emph{A
  perfect-information construction for coordination in games}, Proceedings of
  {F}oundations of {S}oftware {T}echnology and {T}heoretical {C}omputer
  {S}cience ({FSTTCS} 2011), LIPIcs, vol.~13, Schloss-Dagstuhl --
  Leibniz-Zentrum f{\"u}r Informatik, 2011, pp.~387--398.

\bibitem{BouyerBMU12}
Patricia Bouyer, Romain Brenguier, Nicolas Markey, and Michael Ummels,
  \emph{Concurrent games with ordered objectives}, Foundations of Software
  Science and Computational Structures {FOSSACS} 2012. Proc., 2012,
  pp.~301--315.

\bibitem{BBMU15}
Patricia Bouyer, Romain Brenguier, Nicolas Markey, and Michael Ummels,
  \emph{Pure {N}ash equilibria in concurrent games}, Logical Methods in
  Computer Science \textbf{11} (2015), no.~2:9.

\bibitem{ChatterjeeEtAl14}
Krishnendu Chatterjee, Laurent Doyen, Emmanuel Filiot, and
  Jean{-}Fran{\c{c}}ois Raskin, \emph{Doomsday equilibria for omega-regular
  games}, Verification, Model Checking, and Abstract Interpretation, LNCS, vol.
  8318, Springer, 2014, pp.~78--97.

\bibitem{CornejoKuhn2010}
Alejandro Cornejo and Fabian Kuhn, \emph{Deploying wireless networks with
  beeps}, Distributed Computing: 24th International Symposium, {DISC} 2010.
  Proceedings (Nancy~A. Lynch and Alexander~A. Shvartsman, eds.), Springer,
  2010, pp.~148--162.

\bibitem{FaginHMV95}
Ronald Fagin, Joseph~Y. Halpern, Yoram Moses, and Moshe~Y. Vardi,
  \emph{Reasoning about knowledge}, MIT Press, 1995.

\bibitem{FinkbeinerSch05}
B.~Finkbeiner and S.~Schewe, \emph{{U}niform distributed synthesis}, Proc. of
  Logic in Computer Science~(LICS'05), IEEE, 2005, pp.~321--330.

\bibitem{FudenbergLM94}
Drew Fudenberg, David~I Levine, and Eric Maskin, \emph{{The Folk Theorem with
  Imperfect Public Information}}, Econometrica \textbf{62} (1994), no.~5,
  997--1039.

\bibitem{GutierrezWoo14}
Julian Gutierrez and Michael Wooldridge, \emph{Equilibria of concurrent games
  on event structures}, Computer Science Logic (CSL) and Logic in Computer
  Science (LICS) (New York, NY, USA), CSL-LICS '14, ACM, 2014, pp.~46:1--46:10.

\bibitem{halpern-csgt}
Joseph~Y. Halpern, \emph{Computer science and game theory: {A} brief survey},
  CoRR \textbf{abs/cs/0703148} (2007).

\bibitem{KandoriHit98}
Michihiro Kandori and Hitoshi Matsushima, \emph{Private observation,
  communication and collusion}, Econometrica \textbf{66} (1998), no.~3, pp.
  627--652.

\bibitem{KupfermanVar01}
Orna Kupferman and Moshe~Y. Vardi, \emph{Synthesizing distributed systems},
  Proc. of LICS~'01, IEEE Computer Society Press, June 2001, pp.~389--398.

\bibitem{MailathSam06}
George Mailath and Larry Samuelson, \emph{Repeated games and reputations:
  Long-run relationships}, Oxford University Press, 2006.

\bibitem{Moses16}
Yoram Moses, \emph{Relating knowledge and coordinated action: The knowledge of
  preconditions principle}, Theoretical Aspects of Rationality and Knowledge,
  {TARK} 2015, Proc., {EPTCS}, vol. 215, 2015, pp.~231--245.

\bibitem{PetersonRei79}
Gary~L. Peterson and John~H. Reif, \emph{{M}ultiple-{P}erson {A}lternation},
  Proc 20th Annual Symposium on Foundations of Computer Science, (FOCS 1979),
  IEEE, 1979, pp.~348--363.

\bibitem{PnueliRos90}
Amir Pnueli and Roni Rosner, \emph{Distributed reactive systems are hard to
  synthesize}, Proceedings of the 31st Annual Symposium on Foundations of
  Computer Science, (FoCS~'90), 1990, pp.~746--757.

\bibitem{RaskinSan14}
Jean-Francois Raskin and Ocan Sankur, \emph{{Multiple-Environment Markov
  Decision Processes}}, Foundation of Software Technology and Theoretical
  Computer Science (FSTTCS 2014), LIPIcs, vol.~29, Schloss Dagstuhl --
  Leibniz-Zentrum f{\"u} Informatik, 2014, pp.~531--543.

\end{thebibliography}

\end{document}